\documentclass[11pt,a4paper]{article}

\usepackage[utf8]{inputenc}

\usepackage{array,amsbsy,amscd,amsfonts,color,amssymb,amstext,amsmath,latexsym,amsthm,amscd
,multicol,graphicx,makeidx,tikz}
\usepackage[english]{babel}
\usepackage{hyperref}

\makeindex

\newtheorem{theorem}{Theorem}
\newtheorem{proposition}[theorem]{Proposition}

\theoremstyle{definition}
\theoremstyle{definition}
\theoremstyle{definition}\newtheorem{remark}[theorem]{Remark}
\theoremstyle{definition}\newtheorem{example}[theorem]{Example}
\theoremstyle{definition}
\theoremstyle{definition}
\theoremstyle{definition}
\theoremstyle{definition}
\theoremstyle{definition}\newtheorem{conclusion}[theorem]{Conclusion}

\def\proofof [#1] {\noindent {\bf Proof of #1. } }

\def\v #1.{\mathord{\raise 3pt\hbox{\mathsurround=0pt $\mathop\vee\limits^{#1}$\mathsurround=5pt}}}
\def\al #1.{{\mathcal{#1}}}

\renewcommand{\H}{\mathcal{H}}
\renewcommand{\L}{\mathcal{L}}
\newcommand{\N}{\mathbb{N}}
\newcommand{\NN}{\mathbb{N}_0}
\newcommand{\R}{\mathbb{R}}
\newcommand{\C}{\mathbb{C}}

\newcommand{\T}{\mathcal{T}}
\newcommand{\EE}{\mathcal{E}}

\newcommand{\id}{\operatorname{id}}

\newcommand{\E}{\mathbb{E}}
\newcommand{\gl}{\mathfrak{gl}}
\renewcommand{\lg}{\mathfrak{g}}
\newcommand{\unit}{\mathbf{1}}
\renewcommand{\t}{\otimes}
\newcommand{\V}{\operatorname{Var}}
\newcommand{\Mat}{\operatorname{M}}
\newcommand{\GL}{\operatorname{GL}}
\newcommand{\tr}{\operatorname{tr }}
\newcommand{\Ad}{\operatorname{Ad }}
\newcommand{\ad}{\operatorname{ad }}

\newcommand{\dom}{\operatorname{dom}}
\newcommand{\rmd}{\operatorname{d}}
\newcommand{\rmi}{\operatorname{i}}

\newcommand{\drift}{\operatorname{drift}}

\newcommand{\rme}{\operatorname{e}}

\newcommand{\Cci}{C^\infty}
\newcommand{\ra}{\rightarrow}
\newcommand{\boxy}{\hfill $\Box$}

\newcommand{\ie}{{i.e.,\/}\ }

\newcommand{\eg}{{e.g.\/}\ }
\newcommand{\cf}{{cf.\/}\ }

\title{A Continuous-Time Diffusion Limit Theorem\\ for Dynamical Decoupling and\\ Intrinsic Decoherence}
\author{
\phantom{X}\\
{\sc Robin Hillier}$^{1}$,
{\sc Christian Arenz}$^2$,
{\sc Daniel Burgarth}$^{2}$\\
\phantom{X}\\
${}^1$ Department of Mathematics and Statistics\\ Lancaster University\\ Lancaster LA1 4YF, UK\\
\phantom{X}\\
${}^2$ Department of Mathematics\\ Aberystwyth University\\ Aberystwyth SY23 2BZ, UK
}

\begin{document}

\maketitle

\begin{abstract}
We discuss a few mathematical aspects of random dynamical decoupling, a key tool procedure in quantum information theory. In particular, we place it in the context of discrete stochastic processes, limit theorems and CPT semigroups on matrix algebras. We obtain precise analytical expressions for expectation and variance of the density matrix and fidelity over time in the continuum-time limit depending on the system Lindbladian, which then lead to rough short-time estimates depending only on certain coupling strengths. We prove that dynamical decoupling does not work in the case of intrinsic (i.e., not environment-induced) decoherence, and together with the above-mentioned estimates this yields a novel method of partially identifying intrinsic decoherence.
\end{abstract}

\section{Introduction}\label{sec:one}

The aim of this article is two-fold: first, to provide an analytical description of random dynamical decoupling because analytical expressions are often more manageable than combinatoric-numerical ones; second, to use this description to propose a partial method of detecting intrinsic decoherence of quantum systems.

Dynamical decoupling is a method applied to stabilise states of quantum registers against undesired time-evolution. Originally invented in NMR technology, it has been generalised to a wider context, in particular in quantum information theory \cite{LB,VKL}. It works by application of repeated instantaneous unitary correction pulses on the quantum register, perturbing the original time-evolution. The procedure is particularly interesting and effective when performed in a random way \cite{vi3,vi4}. 

While several general estimates and specialisations of the procedure have been proposed in the past (\cf \cite{LB} and the references therein), our focus here is on finding handy analytical descriptions of the time-evolution of expectation and distribution of physically interesting quantities like the density matrix process or the gate fidelity process arising from this random time-evolution. We believe that such descriptions are a valuable tool in future computations and enable predictions in experiments.

We would like to provide now a rough overview of the content of this article. Let $(\H,\L)$ stand for a generic finite-dimensional quantum system, with $\H$ a finite-dimensional complex Hilbert space and $\L$ a possibly time-dependent Lindblad generator (\cf \cite{pb,wo} for general background information). We start in Section \ref{sec:two} by introducing dynamical decoupling and show that the decoupling condition \eqref{eq:deccond2} can be satisfied only if $\L=\rmi [H,\cdot]$ for some Hamiltonian $H$. This might sound like a contradiction since dynamical decoupling aims to eliminate decoherence (noise) arising from open systems. It can be resolved by differentiating between \emph{intrinsic} and \emph{extrinsic decoherence}: the latter one is where decoherence arises from interaction with an actual quantum heat bath or environment such that the total space time-evolution is unitary, the former one is where decoherence is actually the time-evolution of a closed system \cite{Adler}. It is unclear whether intrinsic decoherence may appear in nature, and it would, of course, contradict the axiom of unitary time-evolution. But in order to find out whether it may exist or whether the axiom of unitarity is always verified, one has to perform experiments and develop mathematical tools.

To this end, Sections \ref{sec:three} and \ref{sec:four} provide a probabilistic-analytical approach to dynamical decoupling, namely: we set up a probabilistic description of a random walk in the completely positive trace-preserving (CPT) maps of the quantum system arising from the random correction pulses; then we study the continuum-limit of this random walk under a suitable scaling, which becomes a Gaussian (Markov) process in the CPT maps. We use this to determine the expectation and higher moments of the density matrix process $\rho_t$. In the \ref{sec:four}th section we then compute the expectation of the gate fidelity, which might be regarded as a mean fidelity when averaging over all states on $B(\H)$ in a suitable manner. We illustrate all constructions and considerations with an easy example that shall accompany us through the paper.

Up to this point, things were quite general, but this is where we can turn to our second aim: distinguishing between intrinsic and extrinsic decoherence (with bounded Hamiltonian dilations, \cf \cite[App.]{ABH}). We therefore specialise in the final Section \ref{sec:five} on providing approximative bounds for the gate fidelity in these two extremal cases together with a recipe which should enable the experimenter to determine the type of decoherence present in his setting. Ideally he should just know the pulse length $\tau$, the total time $t$ of evolution, and the coupling strength of the undesired decoherence. In some cases unfortunately some further input is needed. However, the overall moral is roughly speaking the following: the rate of decoherence decreases to $0$ when $\tau\ra 0$ if decoherence is extrinsic and due to bounded interaction \cite{vi3,ABH}; it remains essentially unaffected by the decoupling procedure if decoherence is intrinsic. In other words, if random dynamical decoupling with $\tau\ra 0$ does not eliminate decoherence, then it was intrinsic or due to unbounded interaction! Model studies and illustrations of this procedure can be found in the companion article \cite{ABH}.

\bigskip

{
\noindent\textbf{Acknowledgements.} We would like to thank Micha\l\  Gnacik, Lorenza Viola and the referee for useful discussions and/or comments on the manuscript. Moreover, RH would like to thank Gernot Alber and Burkhard K\"ummerer for guidance in his master thesis several years ago, in which a special case of Section \ref{sec:three} had been developed. 
}

\section{The concept of dynamical decoupling}\label{sec:two}

Let us start with some notation used throughout the article. We shall denote our quantum system in question by $(\H,\L)$, with $\H$ a finite-dimensional complex Hilbert space (of dimension $d_\H$) and we denote the adjoint of linear maps on it by `$^*$'; $\L$ is the Lindblad operator on $B(\H)$, generating the completely positive trace-preserving (CPT) time evolution maps $\alpha_t=\rme^{t\L}$, $t\in\R_+$, of the quantum system (\cf \cite{pb,wo} for general background information). Let us abbreviate $A:=B(\H)$, which has dimension $d=d_\H^2$ and which becomes a Hilbert space again with scalar product $(x,y)\in A\times A \mapsto \langle x, y\rangle := \tr(x^*y)$, and we denote adjoints of maps on this Hilbert space by `$^\dagger$'. We write $\Ad$ (or $\ad$) for the adjoint representation of the unitary group (or its Lie algebra, respectively) on $A$, \ie $\Ad(v)(x)=v x v^*$ and $\ad(H)(x)=[H,x]$, for $v,H,x\in A$ with $v$ unitary and $H$ selfadjoint. 

By a \emph{decoupling set} in $A$ we mean a finite group of unitaries $V:=(v_j)_{j\in J}\subset A$ such that $0\in J$ and
\begin{equation}\label{eq:deccond1}
v_0=\unit, \quad \sum_{j\in J} \Ad(v_j)(x) \in \C\unit, \quad \forall x\in A.
\end{equation}
Notice that in this case we automatically have $\frac{1}{|J|}\sum_{j\in J} v_j x v_j^*=\frac{1}{d_\H}\tr(x)\unit$.

\begin{example}\label{ex:1}
The standard illustrative example of a finite-dimensional quantum system to keep in mind throughout this paper is an $N$-qubit quantum system, so $\H=(\C^2)^{\otimes N}$ and $d=2^{2N}$; there typically the decoupling set $V$ consists of the $4^N$ different combinations of Pauli matrices $\{\unit, \sigma_1,\sigma_2,\sigma_3\}$ on the tensor factors. Here $v_j^*=v_j$., for all $j$.
\end{example}

Given the CPT semigroup of time evolution maps $(\alpha_t)_{t\in\R_+}$ of our system and a (``short") time $\tau$, consider the externally modified time evolution
\begin{equation}\label{eq:randomwalk}
\alpha^{(\tau)}_{(n+1)\tau}= \Ad(v_{j_0}^* v_{j_n})\circ\alpha_\tau \circ \Ad(v_{j_n}^*v_{j_{n-1}})\circ \alpha_\tau\circ \ldots \circ \alpha_\tau \circ\Ad(v_{j_1}^*v_{j_0}),
\end{equation}
where $n\in\N$, and $(j_i)_{i\in\NN}$ forms a certain sequence in $J$ with $j_0 =0$, meaning we apply \emph{instantaneous decoupling} or \emph{correction pulses} $v_{j_i}^*v_{j_{i-1}}$ at time $i\tau$; set $\alpha^{\tau}_t=\alpha_{t-n\tau}\circ\alpha^{\tau}_{n\tau}$ whenever $t\in [n\tau,(n+1)\tau)$. The sequence $j_i$ can be fixed or random, leading to \emph{deterministic} or \emph{random dynamical decoupling}. It turns out that random decoupling has many advantages \cite{vi3,vi4,kas} and moreover is mathematically more interesting, and that is why we want to investigate it here.

Our first goal is to find an analytical description of the externally modified time evolution $\alpha^{(\tau)}_t$. In the random setting, $(\alpha^{(\tau)}_t)_{t\in\tau\N}$ becomes a stochastic process (a random walk with steps lasting time $\tau$) induced by the process $(v_{j_t})_{t\in\tau\N}$ with independent identically distributed (iid) and equidistributed \cite{Fel} increments in $V$, and we are interested in the limit $\tau\ra 0$, which would enable nice analytical expressions.

Since $\alpha_t = \exp( t\L)$, we find 
\begin{equation}\label{eq:v-alpha}
\alpha^{(\tau)}_{(n+1)\tau}= \exp(\tau \Ad(v_{j_n})\circ \L \circ \Ad(v_{j_n}^*)) \circ \alpha^{(\tau)}_{n\tau},
\end{equation}
so the increment during the time interval $(t-\tau,t)$ is given by $\exp(\tau \Ad(v_{j_t})\circ \L \circ \Ad(v_{j_t}^*))$.

We say that $V$ satisfies the \emph{decoupling condition} for $(\H,\L)$ if
\begin{equation}\label{eq:deccond2}
\sum_{j\in J} \Ad(v_j)\circ \L \circ \Ad(v_j^*)=0.
\end{equation}
The idea behind this condition is that it ensures cancellation of interaction at first order in $\tau\|\L\|$, \ie for short time $\tau$, and thus higher order terms contribute.

We say (the time evolution of) our quantum system $(\H,\L)$ is \emph{purely unitary} if $\L=\rmi \ad(H)$, with $H\in A$ selfadjoint, because in this case $\alpha_t=\rme^{t\L}$ is induced by a one-parameter family of unitary matrices; in this case $\L^\dagger=-\L$. The ``opposite case'', namely where $\L^\dagger=\L$ we call \emph{purely dephasing}.
 
\begin{theorem}\label{lem:deccond-unitary}
A decoupling set for $(\H,\L)$ satisfies the decoupling condition \eqref{eq:deccond2} iff $(\H,\L)$ is purely unitary.
\end{theorem}

\begin{proof}
We write out the generator $\L$ in Christensen-Evans form \cite{ce, wo}
\[
\L(x) = \Psi(x) + a x + x a^*, \quad \forall x\in A, 
\]
with a certain $a\in A$ and completely positive $\Psi$ which is not a multiple of $\id_A$ (w.l.o.g., because adding $2\lambda \id_A$ to $\Psi$ has the same result on $\L$ as adding $\lambda \unit$ to $a$). Suppose first that $\alpha$ is purely unitary; then $\Psi=0$ and $a^*=-a$. From \eqref{eq:deccond1}, with $J$ indexing the decoupling set $V=\{v_j:j\in J\}$ as above, we obtain
\begin{align*}
\frac{1}{|J|}\sum_{j\in J} \Ad(v_j)\circ \L \circ\Ad(v_j^*)(x) 
=& \frac{1}{|J|}\sum_{j\in J} \Ad(v_j)(a \Ad(v_j^*)(x)+\Ad(v_j^*)(x)a^*)\\ 
=& \frac{1}{|J|} \sum_{j\in J} (\Ad(v_j)(a) x + x \Ad(v_j)(a^*))\\
=& \frac{1}{d_\H}\tr(a+a^*)x =0,\quad x\in A, 
\end{align*}
so $V$ satisfies the decoupling condition. If instead $\alpha$ is not purely unitary, we have $\Psi\not=0$ and hence
\[
\Phi:= \frac{1}{|J|}\sum_{j\in J} \Ad(v_j)\circ \Psi \circ \Ad(v_j^*)
\]
is completely positive and nonzero. Suppose $\Phi(x)$ equals
\[
- \frac{1}{|J|} \sum_{j\in J} (\Ad(v_j)(a \Ad(v_j^*)(x))+\Ad(v_j)(\Ad(v_j^*)(x)a^*)) = - \frac{1}{d_\H}\tr(a+a^*) x, \quad x\in A.
\]
Then, for every rank-one projection $p\in A$, we have $\Phi(p)=-\frac{1}{d_\H}\tr(a+a^*) p$. But for every $\xi\in\H$ with $p\xi=0$, we have
\[
0= -\frac{1}{d_\H}\tr(a+a^*) \langle \xi, p\xi\rangle = \langle \xi, \Phi(p)\xi\rangle
=  \frac{1}{|J|} \sum_{j\in J} \langle \xi, \Ad(v_j)\circ\Psi\circ \Ad(v_j^*)(p)\xi\rangle,
\]
with each single term $\ge 0$, due to the positivity of $\Psi$ and the scalar product, and thus actually $=0$. In particular, since $v_0=\unit$, we have $\langle \xi, \Psi(p)\xi\rangle=0$, so $\Psi(p) \in \R_+p$.  Let us write $\Psi$ in (minimal) Kraus form with $\operatorname{rank}(\Psi)$ its Kraus rank and certain $b_i\in A$ \cite{wo}:
\[
\Psi(x) = \sum_{i=1}^{\operatorname{rank}(\Psi)} b_i x b_i^*, \quad x\in A.
\] 
This entails then, for any two mutually orthogonal vectors $\eta,\xi\in H$,
\[
0 = \langle \xi, \Psi(|\eta\rangle\langle\eta |) \xi \rangle
= \sum_{i=1}^{\operatorname{rank}(\Psi)} | \langle \xi, b_i \eta \rangle |^2.
\]
Hence, for every $i$, we see that $b_i \eta\in \C\eta$, or in other words, $\eta$ must be an eigenvector of $b_i$. This holds for every $\eta\in\H$, so $b_i\in \C \unit$, and thus $\Psi$ is a multiple of $\id_A$, which contradicts our initial assumptions. Therefore, $\Phi(x)\not= -\frac{1}{d_\H}\tr(a+a^*) x$, and
\[
\frac{1}{|J|}\sum_{j\in J} \Ad(v_j)\circ \L \circ \Ad(v_j^*)\not= 0,
\]
\ie $V$ on $(\H,\L)$ does not satisfy the decoupling condition.
\end{proof}

Despite this result, it will turn out in the course of this paper that dynamical decoupling is still interesting beyond the unitary case.

\section{The continuous-time limit of random dynamical decoupling}\label{sec:three}

We continue with the notation and concepts introduced in the previous section. Let us, in particular, assume all increments $v_{j_i}$, with $i\in\N$, in our random walk $(v_{j_i})_{i\in\NN}$ of decoupling pulses to be iid and equidistributed in $V$ as in the preceding section. The induced random walk $(\alpha^{(\tau)}_{(n+1)\tau})_{n\in\N}$ lies in the completely positive maps on $A$ according to \eqref{eq:randomwalk} and \eqref{eq:v-alpha}. Moreover, since completely positive maps are linear maps of the Hilbert space $A$ and since all the increments are invertible, the random walk actually lies in the group $\GL(A)$ of invertible linear maps of $A$, and $\L\in\gl(A)$, the Lie algebra of $\GL(A)$. This induced random walk has again iid increments and is described by the measure
\begin{equation}\label{eq:mutau}
\mu ^{(\tau)}:=\frac{1}{|J|} \sum_{j\in J} \delta
_{\exp(\tau \Ad(v_j)\circ \L \circ \Ad(v_j^*))}.
\end{equation}
We would like to investigate it in the limit $\tau\ra 0$. However, since $\tau$ is an actual physical quantity in our set-up, we keep it and instead consider a fictitious limit, which should be good for small $\tau$, as explained below. Considering simply $\mu^{(\tau)}$ and the limit of $\tau\ra 0$, we would obtain a drift-like expression without fluctuations, which is not a really physical result but a good first approximation, \cf \eqref{eq:mu-n2} and Remark \ref{rem:drift-limit}. In fact, the well-known Donsker invariance principle \cite{Fel} basically says that the limit of a classical random walk is described suitably well by Brownian motion, \ie by scaling length increments with the square-root of time increments, supposed that the expectation of every increment is $0$. The following kind of central limit theorem helps us to treat these dissipation-fluctuation terms in the present noncommutative setting, and will thus become a building stone in our construction; it has first been stated in \cite{we} but can also be found in the textbook \cite[Th.4.4.2]{gr}. The necessary notation and concepts in Lie groups and stochastic processes on Lie groups are lined out in Appendix \ref{sec:app}, and we suggest the reader to go through it before continuing here.

\begin{theorem}\label{theowehn}
Let $G$ be an $N$-dimensional Lie group, with $\unit$-chart $(U,x)$, Lie algebra basis $(X_k)_{1\le k\le N}$ and coordinate mappings $x_k:U\ra \R $ extended to functions in $\Cci_c(G)$ and hence to the one-point compactification $G_c$.
Let $(\mu_n)_{n\in \N}$ be a family of probability measures on $G$ converging to $\delta _{\unit}$. Suppose there are numbers $a_k,a_{kl} \in \R$ such that $(a_{kl})_{k,l=1\ldots N}$ is positive semi-definite and, for all $k,l=1,...,N$ and $n\ra\infty$:
\begin{itemize}
\item[\textup{(i)}] $ \int_{G} x_k(g) d\mu _n(g) = a_k/n + o(1/n),$
\item[\textup{(ii)}] $\int_{G} x_k(g)x_l(g) \rmd\mu _n(g) = a_{kl}/n + o(1/n),$
\item[\textup{(iii)}] $\mu _n(\tilde{U}^c)=o(1/n)$ for every $\unit$-neighbourhood
  $\tilde{U}\subset G_c$.
\end{itemize}
Then the sequence $((\mu _n) ^{*n})_{n \in \N}$ converges *-weakly to a measure
$\nu_1$ on $G_c$ which belongs to the convolution semigroup $(\nu_t)_{t \in
  \R _+}$ whose corresponding operator semigroup $(T_t)_{t\in\R_+}$ on $C(G_c)$ has infinitesimal generator 
\[L:= \frac{\rmd}{\rmd t} T_t\restriction_{t=0} 
=\sum_{k=1}^N a_k D_{X_k} + \sum_{k,l=1}^N a_{kl}D_{X_k}
D_{X_l}\]
with $\dom(L)= C^2(G_c)$.
\end{theorem}

We would like to apply this theorem to our setting, namely where $G=\GL(A)$ and $\lg=\gl(A)$ regarded as (real!) linear Lie group and algebra, respectively, and subspaces of $B(A)$. 
The plan is as follows: in a first step we shall construct a continuous-time stochastic process in $G$, and in a second step use this to obtain a description of the induced behaviour of the density matrix process $(\rho_t)_{t\in\R_+}$.
In our setting this means we first have to define suitable and physically realistic measures $\mu_n$ to which to apply our limit procedure of Theorem \ref{theowehn}. The drift part should correspond to the original drift part resulting from \eqref{eq:mutau}. Putting
\[
\bar{\L} := \frac{1}{|J|}\sum_{j\in J}\Ad(v_j)\circ \L\circ\Ad(v_j^*)
\]
and
\[
\L_j:= \Ad(v_j)\circ (\L-\bar{\L})\circ\Ad(v_j^*),
\]
let us define the measures
\begin{equation}\label{eq:mu-n}
\mu_n := \frac{1}{|J|}\sum_{j\in J}\delta
_{\exp\big(\frac{\tau}{n^{1/2}} \L_j + \frac{\tau}{n} (\bar{\L} - \frac{\tau}{2} \L_j^2)\big)}, \quad n\in\N,
\end{equation}
which conceptually imitate a diffusion part for the variation around the mean $\bar{\L}$ and a drift part for the mean movement. Apart from being mathematically clear and plausible from the classical Donsker invariance principle, the meaningfulness of this limit shall moreover be confirmed by numerical analysis carried out partially in the final section and mainly in \cite{ABH}. 

Let us drop a quick side remark: as a rough first approximation for $\mu_n$ we might also study the purely drift-like
\begin{equation}\label{eq:mu-n2}
\mu_n^{(\drift)} := \frac{1}{|J|}\sum_{j\in J}\delta
_{\exp(\frac{\tau}{n} \Ad(v_j)\circ \L \circ \Ad(v_j^*))},
\end{equation}
similar to \eqref{eq:mutau} but with scaling variable $\tau/n$ instead of $\tau$ as we now would like to keep $\tau$ fixed. In analogy to the law of large numbers this would lead to even nicer expressions and CPT dynamics but less faithful modelling; sometimes we will consider it briefly for comparison reasons, \cf Remark \ref{rem:drift-limit}. It can be shown that all other types of scaling (\ie others than $1/n$ or $1/\sqrt{n}$) essential lead to either trivial or singular not well-defined expressions. We shall therefore stick to $\mu_n$ henceforth if not explicitly mentioned otherwise. Moreover, it is clear that $\Ad(v_j)\circ\bar{\L}\circ\Ad(v_j^*) = \bar{\L}$ (as $V$ is a group) and hence $\sum_{j\in J} \L_j = 0$. We have $\bar{\L}=0$ iff $V$ satisfies the decoupling condition for $(\H,\L)$.


Using now the defining property of the coordinate maps in the limit $n\ra\infty$, (i) is obtained with a series expansion of $\exp\big(\frac{\tau}{n^{1/2}} \L_j + \frac{\tau}{n} (\bar{\L} - \frac{\tau}{2} \L_j^2)\big)$, namely:
\begin{align*}
a_k=& \lim_{n\ra\infty} \frac{n}{\tau} \int_{G_c} x_k(g) \rmd \mu_n(g) \\
=& \lim_{n\ra\infty} \frac{n}{\tau|J|} \sum_{j\in J} x_k\Big(\exp\big(\frac{\tau}{n^{1/2}} \L_j + \frac{\tau}{n} (\bar{\L} - \frac{\tau}{2} \L_j^2)\big)\Big)\\
=& \lim_{n\ra\infty} \frac{n}{\tau |J|} \sum_{j\in J} \Big( 
\frac{\tau}{n^{1/2}}\langle \L_j, X_k\rangle_\lg + \frac{\tau}{n}\langle \bar{\L} - \frac{\tau}{2}\L_j^2, X_k \rangle_\lg + \frac{\tau^2}{2 n}\langle \L_j^2, X_k\rangle_\lg + O\big(\frac{1}{n^{3/2}}\big)\Big)\\
=& \langle \bar{\L}, X_k \rangle_\lg.
\end{align*}
Analogously, for (ii) we have
\begin{align*}
a_{kl}=& \lim_{n\ra\infty} \frac{n}{\tau} \int_{G_c} x_k(g)x_l(g) \rmd \mu_n(g) \\
=& \lim_{n\ra\infty} \frac{n}{\tau |J|}\sum_{j\in J} 
x_k\Big(\exp\big(\frac{\tau}{n^{1/2}} \L_j + \frac{\tau}{n} (\bar{\L} - \frac{\tau}{2} \L_j^2)\big)\Big)x_l\Big(\exp\big(\frac{\tau}{n^{1/2}} \L_j + \frac{\tau}{n} (\bar{\L} - \frac{\tau}{2} \L_j^2)\big)\Big)\\
=& \lim_{n\ra\infty} \frac{n}{\tau |J|}\sum_{j\in J}\Big(
\frac{\tau^2}{n}\langle \L_j, X_k\rangle_\lg\langle \L_j, X_l\rangle_\lg + O\big(\frac{1}{n^{3/2}}\big)\Big)\\
=&  \frac{\tau}{|J|}\sum_{j\in J} \langle \L_j, X_k\rangle_\lg\langle \L_j, X_l\rangle_\lg,
\end{align*}
for every $k,l=1,\ldots ,N$.
Finally, it is easy to see that condition (iii) is satisfied as $\mu_n$ has discrete support in $|J|$ points only which converge to $0$ as $n\ra\infty$. Thus we get
\begin{equation}\label{eq:CLTresult}
L= \sum_{k=1}^N a_k D_{X_k} + \sum_{k,l=1}^N a_{kl} D_{X_k}D_{X_l} = D_{\bar{\L}} + \frac{\tau}{|J|}\sum_{j\in J} D_{\L_j}^2
\end{equation}
for the generator of the limit convolution semigroup $(\nu_t)_{t\in\R_+}$ on $G_c$, which can be interpreted as a combination of drift and diffusion on $G_c$. This has been the first big step in our construction, namely the construction of the convolution semigroup of measures $(\nu_t)_{t\in\R_+}$ on $G$; it implicitly describes a stochastic process $(\alpha'_t)_{t\in\R_+}$ on $G$ (according to Theorem \ref{theowehn}) with $\alpha_0'=\id_A$. 

Our second step shall be to calculate the time evolution of the density matrix and related physically significant quantities out of the stochastic process $(\alpha_t')_{t\in\R_+}$. This is slightly involved, but can be done using some tools which we are now going to derive.

A general fact is that, for every $f\in C(G_c)$, we have
\begin{equation}\label{eq:CLTexp}
\E [ f\circ\alpha'_t]
= \int_{G_c} f(g) \rmd \nu_t (g) 
= T_t f (\unit).
\end{equation}
We define the subsemigroup $G^{[1]}:=\{g\in G: \|g\|\le 1\}\subset G\subset G_c$.
Since $\alpha^{(\tau)}_t$ are contractions, the measures $\mu_n$ must all be supported in $G^{[1]}$. This implies that the convolutions $\mu_n^{*m}$ of those measures are supported in $G^{[1]}$ (\cf Appendix \ref{sec:app} for a proof). For given $t>0$, choosing a sequence $(m_n)_{n\in\N}$ such that $m_n/n\ra t$, one can check that $\mu_n^{*m_n}\ra \nu_t$. Hence the limit semigroup $(\nu_t)_{t\in\R_+}$ is supported in $G^{[1]}$, meaning the Gaussian process $(\alpha'_t)_{t\in\R_+}$ stays almost surely in  $G^{[1]}$. Then it follows that $T_t$ preserves the closed subspace $C_{0,b}(G^{[1]c})\subset C(G_c)$ of bounded continuous functions on the complement $G^{[1]c}$  of $G^{[1]}$ vanishing at the boundary $\{g\in G: \|g\|=1\}$:
\[
T_t f(g) = \int_{G^{[1]}} f(hg) \rmd \nu_t (h) = 0,
\]
for every $g\in G^{[1]}$, as $f(hg)=0$ for $h,g\in G^{[1]}$, \ie $T_tf$ has support in $G^{[1]c}$ and is bounded by $\|f\|_\infty$.
The corresponding quotient Banach space $C(G_c)/C_{0,b}(G^{[1]c})$ can be identified with $C_b(G^{[1]})$: namely, $f\in C(G_c)$ induces a function $f\restriction_{G^{[1]}}\in C_b(G^{[1]})$ and v.v., two extensions $f_1,f_2$ of a function $f\in C_b(G^{[1]})$ to $G_c$ lead to $f_1-f_2\in C_{0,b}(G^{[1]c})$, thus a unique element in $C(G_c)/C_{0,b}(G^{[1]c})$. Write $q$ for the corresponding quotient map and $f^{[1]}:= q(f)$, for every $f\in C(G_c)$, so that $f^{[1]}(g)= f(g)$ if $g\in G^{[1]}$. Then we get the quotient semigroup  $(T_t^{[1]})_{t\in\R_+}$ as in Appendix \ref{sec:app}, with infinitesimal generator $K= q(L q^{-1}( \cdot))$ and $\dom(K) = q(\dom(L)) \simeq C_b^2(G^{[1]})$, the twice differentiable functions on $G^{[1]}$ which and whose first and second order derivatives are all bounded.

In order to achieve a description of the time evolution $(\rho_t)_{t\in\R_+}$ of the density matrix, the idea is to study every entry of $\rho_t$ in a certain orthonormal basis. To this end, let $(e_k)_{k=1...d}$ be an arbitrary fixed orthonormal basis of $A$. We consider, for every $k,l$, the function
\[
f_{kl}: g \in G \mapsto \langle e_l, g (e_k)\rangle,
\] 
which is $\nu_t$-integrable (because bounded by $1$ on the support) and which lies in $\Cci(G)$ but not in $C_c(G)$. Write $f_{kl}^{[\infty]}$ for an arbitrary but fixed function in $\Cci_c(G)$ (and hence $\Cci(G_c)$) coinciding with $f_{kl}$ on $G^{[1]}$, which can always be achieved, \eg by multiplying with a smoothed indicator function on $G^{[1]}$ (easy exercise); moreover, following the notation of the preceding paragraph we write $f_{kl}^{[\infty,1]}:= q(f_{kl}^{[\infty]})$. Then
\[
\E[f_{kl}(\alpha'_t(\cdot)g)]=\E[f_{kl}^{[\infty]}(\alpha'_t(\cdot)g)]
= T_t f_{kl}^{[\infty]} (g) = T_t^{[1]} f_{kl}^{[\infty,1]} (g), \quad t\in\R_+, g\in G^{[1]}.
\]
Noticing furthermore that $f_{kl}^{[\infty,1]}\in C_b^2(G^{[1]})$, we have, for every $g\in G^{[1]}$,
\begin{equation}\label{eq:LL}
\begin{aligned}
K f_{kl}^{[\infty,1]}(g) =& q(L f_{kl}^{[\infty]})(g) 
= L f_{kl}^{[\infty]}(g)\\
=& \frac{1}{|J|}\sum_{j\in J} \Big( D_{\bar{\L}}+ \tau D_{\L_j}^2\Big) f_{kl}^{[\infty]}(g)\\
=& \frac{\rmd}{\rmd t}\langle e_l,\rme^{t\bar{\L}}g (e_k) \rangle \restriction_{t=0}
+\frac{\tau}{|J|}\sum_{j\in J} \frac{\rmd^2}{\rmd t^2}\langle e_l,\rme^{t\L_j}g (e_k) \rangle \restriction_{t=s=0} \\
=& \langle e_l , \Big(\bar{\L}
+\frac{\tau}{|J|}\sum_{j\in J} \L_j^2 \Big)(g e_k) \rangle \\
=& \langle e_l, \hat{L}(g e_k)\rangle 
\end{aligned}
\end{equation}
with
\[
\hat{L} := \bar{\L} + \frac{\tau}{|J|} \sum_{j\in J} \L_j^2 \in B(A).
\]
Analogously $K^n f_{kl}^{[\infty,1]}(g) =\langle e_l, \hat{L}^n(g e_k)\rangle$, which is bounded by $\|\hat{L}\|^n$ uniformly in $g\in G^{[1]}$. Therefore, \[
z\in \C \mapsto  \sum_{n=0}^\infty \frac{z^n}{n!} K^n f_{kl}^{[\infty,1]}
= \sum_{n=0}^\infty \frac{z^n}{n!} \langle e_l, \hat{L}^n (\cdot e_k)\rangle = \langle e_l, \rme^{z \hat{L}} (\cdot e_k)\rangle \in \Cci(K)
\]
converges and is an analytic continuation of $t\mapsto T_t^{[1]}f_{kl}^{[\infty,1]}$, so $f_{kl}^{[\infty,1]}\in\Cci(K)$ is an entire analytic vector for $T_t^{[1]}$ (\cf Appendix \ref{sec:app}).

Recalling \eqref{eq:CLTexp} and noticing that $\frac12\unit\in G^{[1]}$ and that $f_{kl}$ is linear in its argument, this enables us to compute the expectation value
\begin{align*}
\E [ \langle e_l,  \alpha'_t(e_k) \rangle]
=& \E [f_{kl}(\alpha'_t(\cdot))] 
= 2 \E[f_{kl}(\frac12\alpha'_t(\cdot))] 
= 2 \E[f_{kl}^{[\infty]}(\frac12\alpha'_t(\cdot))] \\
=& 2 T_t f_{kl}^{[\infty]}\Big(\frac12 \unit\Big)
= 2 T_t^{[1]} f_{kl}^{[\infty,1]} \Big(\frac12 \unit\Big)
= \langle e_l, \rme^{t \hat{L}} (e_k) \rangle,
\end{align*}
for every $k,l$, so $\E[\alpha'_t(e_k)]= \rme^{t\hat{L}}(e_k)$. Since this holds for every basis vector $e_k$, it holds for all elements in $A$. Applying it to the $A$-valued ``density matrix stochastic process" $(\rho_t:=\alpha'_t(\rho_0))_{t\in\R_+}$, we find
\[
\E [\rho_t] = \rme^{t \hat{L}}(\rho_0),
\]
concluding our second step, too.

We summarize this all in

\begin{theorem}\label{th:cont-limit}
The continuous-time limit $(\alpha'_t)_{t\in\R_+}$ of the above random walk determined by a quantum system $(\H,\L)$ and random dynamical decoupling with decoupling set $V=(v_j)_{j\in J}$ and \eqref{eq:mu-n} leads to a contraction semigroup with generator \eqref{eq:CLTresult}. The density matrix $(\rho_t)_{t\in\R_+}$ is then a stochastic process in $A$ with expectation
\[
\E[\rho_t] = \rme^{t \hat{L}}(\rho_0), \quad \forall t\ge 0,
\]
where
\[
\hat{L} = \bar{\L} + \frac{\tau}{|J|} \sum_{j\in J} \L_j^2.
\]
\end{theorem}

\begin{remark}\label{rem:non-const-evo}
If the intrinsic time evolution is not constant (but still continuously differentiable), then the continuous-time limit can be carried out in the same way, resulting in a time-dependent generator 
\[
\hat{L}(t) = \bar{\L}(t) + \frac{\tau}{|J|} \sum_{j\in J} \Ad(v_j)\circ (\L(t)-\bar{\L}(t))^2 \circ \Ad(v_j^*), \quad \forall t\in\R_+,
\]
and just a time-ordered integral \cite{LB}
\begin{equation}\label{eq:time-order}
\E [\rho_t] = \mathcal{T} \rme^{\int_0^t \hat{L}(t')\rmd t'}(\rho_0)
= \sum_{n=0}^\infty \int_0^{t}\int_0^{t'_{n}} \ldots \int_0^{t'_2}
\hat{L}(t'_n)\ldots \hat{L}(t'_1) \rmd t'_1 \ldots \rmd t'_n
\end{equation}
instead of the semigroup. However, this analytic expression will be a good approximation of the original random walk usually only if $\tau$ is sufficiently small such that
\[
\tau \Big\| \frac{\rmd}{\rmd t'} \L(t)\Big\| \ll \|\L(t')\|, \quad \forall t'\in [0,t]. 
\]
For simplicity we shall only deal with the time-independent version here below.
\end{remark}

\begin{remark}\label{rem:drift-limit}
Let us write $\hat{L}^{(\drift)}$ for the  generator and $(\nu_t^{(\drift)})_{t\in\R_+}$ for the convolution semigroup of measures corresponding to the drift-like continuous-time limit of the random walk with $\mu_n^{(\drift)}$ as in \eqref{eq:mu-n2} instead of \eqref{eq:mu-n}, and accordingly $\E^{(\drift)}$ and $\V^{(\drift)}$ for expectation and variances with respect to $(\nu_t^{(\drift)})_{t\in\R_+}$. Then going through the construction of Theorem \ref{th:cont-limit}, we see that the generator of $(T_t^{(\drift)})_{t\in\R_+}$ becomes $L^{(\drift)}=D_{\bar{\L}}$. Hence $\hat{L}^{(\drift)}= \bar{\L}$, which vanishes iff the decoupling condition is fulfilled iff the original time evolution $\alpha$ was unitary, according to Theorem \ref{lem:deccond-unitary}. In this case $T_t^{(\drift)}=\id$, for all $t\in\R_+$, hence $\E^{(\drift)}[\rho_t]=\rho_0$.
\end{remark}

\begin{example}\label{ex:2}
(1) We continue our Example \ref{ex:1} from the preceding section, the $N$-qubit system, with $V$ the group of tensor products of $N$ Pauli matrices. Suppose our time evolution is unitary, so $\L= \rmi \Ad(H)$ with $H$ the system Hamiltonian. Then we find $\hat{L}^{(\drift)}=\bar{\L}=0$, so $\L_j= \rmi [v_j H v_j,\cdot]$ and
\[
\hat{L} = -\frac{\tau}{|J|} \sum_{j\in J} [v_jH v_j,[v_j H v_j, \cdot]].
\]
Now a variety of special cases may be investigated. If \eg $H$ acts only on the first qubit, \ie it can be written as $H=H_1 \otimes \unit^{\otimes (n-1)}$, then so does $\hat{L}$. If moreover $\rho_0$ splits as a product state on the tensor factors, then so does $\E[\rho_t]$, for all $t>0$, with only the first tensor factor changing over time.

(2) Another example, which shall turn up in Figure \ref{fig1} and which is treated in detail in \cite{ABH} is the amplitude-damping model. In this setting $\H$ is the one-qubit Hilbert space $\C^2$, $A= \Mat_2(\C)$ and 
\[
\L(x)= -\gamma \big( 2x - \rmi \sigma_3 x - \rmi x \sigma_3 - \sigma_1 x \sigma_1 
-\sigma_2 x \sigma_2 - \rmi \sigma_1 x \sigma_2 - \rmi \sigma_2 x \sigma_1 \big), \quad x\in A,
\]
with a certain coefficient $\gamma \in \R_+$. The Pauli matrices constitute the decoupling set $V=\{v_0=\unit, v_j=\sigma_j: j=1,2, 3\}$. In order to compute the generator $\hat{L}= \bar{\L} + \frac{\tau}{4} \sum_{j=0}^3 \L_j^2$, one checks:
\[
\bar{\L}(x) = - \gamma \big( 2x - \sigma_1 x \sigma_1 - \sigma_2 x \sigma_2 \big)
\]
and
\[
\L_0(x)= -\L_1(x)= -\L_2(x) = \L_3(x) 
= -\rmi\gamma \big( \sigma_3 x + x \sigma_3 + \sigma_1 x\sigma_2 - \sigma_2 x \sigma_1 \big), \quad x\in A.
\]
A computer can now easily calculate $\E[\rho_t]= \rme^{t\hat{L}}(\rho_0)$, for any given $t>0$ and initial density matrix $\rho_0\in A$. The result should be a good approximation for the actual random walk if $\tau\ll 1/\|\L\|$.
\end{example}

Important related quantities like gate fidelity shall be computed in the following section. Before concluding the present section let us derive here a tool that shall allow us to compute higher moments (including variance) of random variables, beyond the present linear ones (expectation value).

\begin{proposition}\label{prop:higher-mom}
In the setting of Theorem \ref{th:cont-limit}, for all $x_1,\ldots,x_n, y_1,\ldots y_n \in A $, let
\[
f_{x_1\ldots x_n,y_1\ldots y_n}(g) := \langle y_1, g(x_1)\rangle \cdots\langle y_n, g(x_n)\rangle = (f_{x_1,y_1}\cdots f_{x_n,y_n})(g), \quad g\in G,
\]
and define the linear operator $\hat{L}^{(n)}$ on $A ^{\t n}$ by
\begin{equation*}
\begin{split}
\hat{L}^{(n)} (x_1 \t ... \t x_n) :=& 
\sum_{l=1}^n x_1\t ...\t \bar{\L}(x_l) \t ...\t x_n\\ 
&+ \frac{\tau}{|J|} \sum_{j\in J} \Big(\sum_{l=1}^n x_1\t ...\t \L_j^2(x_l) \t ...\t x_n\\  
& \quad + 2 \sum_{k=1,l>k}^{n} x_1 \t ...\t \L_j(x_k)\t ...\t
 \L_j(x_l) \t ...\t x_n \Big)
\end{split}
\end{equation*}
and linear extension. Then 
\[ 
\E [f_{x_1 \ldots x_n,y_1 \ldots y_n}\circ\alpha_t'] 
= \langle(y_1\t ...\t y_n), \rme^{t\hat{L}^{(n)}} (x_1\t ...\t x_n)\rangle.
\] 
\end{proposition}

\begin{proof}
Following the notation and the truncation and quotient space procedure exactly as in the case of $f_{kl}$, we can define (not uniquely) a smooth function $f_{x_1 \ldots x_n,y_1 \ldots y_n}^{[\infty]} \in \Cci(G)$ from $f_{x_1 \ldots x_n,y_1 \ldots y_n}$ and hence a function $f^{[\infty,1]}_{x_1 \ldots x_n,y_1 \ldots y_n}\in C_b^2(G^{[1]})$, which is analytic for $K$, \ie in $\Cci(K)$; we can and do choose it such that $f^{[\infty,1]}_{x_1 \ldots x_n,y_1\ldots y_n}=f^{[\infty]}_{x_1,y_1} \cdots f^{[\infty]}_{x_n,y_n}$. Exploiting then the product rule for differentiation, we obtain 
\begin{equation*}
\begin{split}
K f^{[\infty,1]}_{x_1 \ldots x_n,y_1\ldots y_n}(\unit)
=& L f^{[\infty]}_{x_1 \ldots x_n,y_1\ldots y_n}(\unit)\\
=& L \big(f^{[\infty]}_{x_1,y_1} \cdots f^{[\infty]}_{x_n,y_n} \big)(\unit)\\
=& D_{\bar{\L}} \big(f^{[\infty]}_{y_1,x_1} \cdots f^{[\infty]}_{y_n,x_n}\big)(\unit)
+ \frac{\tau}{|J|}\sum_{j \in J} D_{\L_j}^2
\big(f^{[\infty]}_{y_1,x_1}\cdots f^{[\infty]}_{y_n,x_n}\big)(\unit)\\
=& \sum_{l=1}^n \langle y_1,x_1\rangle \cdots \langle y_l,\bar{\L}(x_l) \rangle \cdots \langle y_n,x_n\rangle \\
&+ \frac{\tau}{|J|}\sum_{j \in J} \Big( \sum_{l=1}^n \langle y_1,x_1\rangle \cdots \langle y_l,\L_j^2(x_l) \rangle \cdots \langle y_n,x_n\rangle \\  
& \quad + \sum_{k=1,l>k}^{n} 2 \langle y_1, x_1\rangle \cdots \langle y_k, \L_j(x_k)\rangle \cdots \langle y_l, \L_j(x_l) \rangle \cdots \langle y_n, x_n\rangle\Big) \\
=& \langle (y_1\t ...\t y_n), \hat{L}^{(n)} (x_1\t ...\t x_n)\rangle.
\end{split}
\end{equation*}

Analogously, for higher powers we have 
\[
K ^k f^{[\infty,1]}_{x_1,...,x_n,y_1,...y_n}(\unit)
= \langle(y_1\t ...\t y_n), (\hat{L}^{(n)})^k (x_1\t ...\t x_n)\rangle
\] 
whence $\exp(t\hat{L}^{(n)})$ is well-defined on $A ^{\t n}$. Thus we find
\[
\E[f_{x_1,...,x_n,y_1,...y_n}\circ \alpha_t'] 
= T^{[1]}_t f^{[\infty,1]}_{x_1,...,x_n,y_1,...y_n} (\unit) = \langle(y_1\t ...\t y_n), \rme^{t\hat{L}^{(n)}} (x_1\t ...\t x_n)\rangle.
\]
\end{proof}

Analogously, one can prove

\begin{proposition}\label{prop:modulus}
In the setting of Theorem \ref{th:cont-limit}, for all $x,x_i,y,y_i \in A $, let
\begin{flalign*}
\check{L}^{(2)} (x \t y) :=& 
\bar{\L} (x)\t y + x\t \bar{\L}^\dagger (y)\\
&+ \frac{\tau}{|J|} \sum_{j\in J} \Big(\L_j^2(x)\t y + 2 \L_j (x)\t \L_j^\dagger(y) + x\t(\L_j^\dagger)^2(y)\Big)
\end{flalign*}
and
\begin{align*}
\check{L}^{(4)} &(x_1\t y_1 \t x_2\t y_2) :=
\bar{\L} (x_1)\t y_1\t x_2\t y_2 + x_1\t \bar{\L} (y_1)\t x_2\t y_2 \\
&+ x_1\t y_1 \t \bar{\L}^\dagger (x_2)\t y_2 +x_1\t y_1 \t x_2\t \bar{\L}^\dagger (y_2)\\ 
& + \frac{\tau}{|J|} 
\sum_{j\in J} 
\Big(\L_j^2(x_1)\t y_1\t x_2\t y_2 + x_1\t \L_j^2(y_1) \t x_2\t y_2 \\
&\quad + x_1\t y_1\t (\L_j^\dagger)^2(x_2) \t y_2 +  x_1\t y_1\t x_2 \t (\L_j^\dagger)^2(y_2) \\
&\quad + 2 (\L_j (x_1)\t y_1 + x_1\t \L_j(y_1))\t (\L_j^\dagger (x_2)\t y_2 + x_2\t \L_j^\dagger(y_2))\\
&\quad +2 \L_j (x_1)\t \L_j(y_1) \t x_2\t y_2 + 2 x_1\t y_1 \t \L_j^\dagger (x_2)\t \L_j^\dagger (y_2) \Big).
\end{align*}
Then
\[
L (|f^{[\infty]}_{x,y}|^2)(\unit) = \langle(y\t x), \check{L}^{(2)} (x\t y) \rangle
\]
and
\[
L (|f^{[\infty]}_{x,y}|^4)(\unit) = \langle(y\t x\t y\t x), \check{L}^{(4)} (x\t y\t x\t y) \rangle.
\]
\end{proposition}

\section{Distribution of the gate fidelity}\label{sec:four}

The most interesting quantity in control theory of a quantum system is its fidelity; as we want to decouple independently of the state, we consider the \emph{gate fidelity} \cite{LB}, which is given by the random variable
\[
F_t := 1- \frac{1}{d} \sum_{k,l=1}^d |\langle e_l, (\id - \alpha'_t)(e_k) \rangle|^2, 
\]
independent of the actual choice of the orthonormal basis $(e_k)_{k=1\ldots d}$ of $A$. Most other versions of fidelity can be treated using similar ideas.
 
We are interested in $\E [F_t]$ and $\V[F_t]$. 

\begin{proposition}\label{prop:EFtVarFt}
In the setting of Theorem \ref{th:cont-limit}, the expectation and variance of the gate fidelity of the quantum system $(\H,\L)$ with decoupling set $V$ are given by
\[
1-\E[F_t] = \frac{1}{d}\sum_{k,l=1}^d
\Big( \delta_{k,l} - \delta_{k,l}  \langle e_l, (\rme^{t\hat{L}}+\rme^{t\hat{L}^\dagger})(e_k) \rangle +\langle e_l\otimes e_k, \rme^{t\check{L}^{(2)}}(e_k\otimes e_l) \rangle \Big)
\]
and
\begin{align*}
\V[F_t]=& \frac{1}{d^2}\sum_{i,j,k,l=1}^d \Big( \delta_{k,l}\delta_{i,j} 
- 2 \delta_{i,j}\delta_{k,l} \langle e_l, (\rme^{t\hat{L}}+\rme^{t\hat{L}^\dagger})(e_k) \rangle \\
&\qquad+2 \delta_{i,j} \langle e_l\t e_k, \rme^{t\check{L}^{(2)}}(e_k\t e_l) \rangle\\ 
&\qquad+  \delta_{i,j}\delta_{k,l} \langle e_i\t e_k, (\rme^{t\check{L}^{(1,1)}}+\rme^{t\check{L}^{(1,2)}} +\rme^{t\check{L}^{(1,2)\dagger}}+\rme^{t\check{L}^{(1,1)\dagger}})
(e_i\t e_k) \rangle \\
&\qquad - 2 \delta_{i,j} \langle e_i\t e_l\t e_k, (\rme^{t\check{L}^{(3,1)}}+\rme^{t\check{L}^{(3,2)}})(e_i\t e_k\t e_l) \rangle \\
&\qquad+ \langle e_j\t e_i\t e_l\t e_k, \rme^{t\check{L}^{(4)}}(e_i\t e_j\t e_k\t e_l) \rangle\Big)\\
&-\frac{1}{d^2}\Big(\sum_{k,l=1}^d \Big(\delta_{k,l} - 2\delta_{k,l}  \langle e_l, (\rme^{t\hat{L}}+\rme^{t\hat{L}^\dagger})(e_k) \rangle + \langle e_l\otimes e_k, \rme^{t\check{L}^{(2)}}(e_k\otimes e_l) \rangle \Big)\Big)^2.
\end{align*}
with
\begin{align*}
\check{L}^{(1,1)}(x\t y) =& 
\bar{\L}(x)\t y+ x\t \bar{\L}(y)\\
&+\frac{\tau}{|J|}\sum_{j\in J} \L_j^2(x)\t y + 2\L_j(x)\t \L_j(y) + x\t \L_j^2(y)\\
\check{L}^{(1,2)}(x\t y) =& 
\bar{\L}^\dagger(x)\t y + x\t \bar{\L}(y)\\
&+\frac{\tau}{|J|}\sum_{j\in J} (\L_j^\dagger)^2(x)\t y + 2\L_j^\dagger(x)\t \L_j(y) + x\t \L_j^2(y)\\
\check{L}^{(3,1)}(x\t y\t z) =& 
\bar{\L}(x)\t y\t z + x\t  \bar{\L}(y) \t z + x\t y\t  \bar{\L}^\dagger(z)\\
&+\frac{\tau}{|J|}\sum_{j\in J} \Big( \L_j^2(x)\t y\t z + 2\L_j(x)\t \L_j(y) \t z + 2\L_j(x)\t y \t \L_j^\dagger(z)\\
&\quad + 2x\t \L_j(y) \t \L_j^\dagger(z) + x\t \L_j^2(y) \t z +x\t y\t (\L_j^\dagger)^2(z) \Big)\\
\check{L}^{(3,2)}(x\t y\t z) =& 
\bar{\L}^\dagger(x)\t y\t z + x\t  \bar{\L}(y) \t z + x\t y\t  \bar{\L}^\dagger(z)\\
&+\frac{\tau}{|J|}\sum_{j\in J} \Big((\L_j^\dagger)^2(x)\t y\t z + 2\L_j^\dagger(x)\t \L_j(y) \t z + 2\L_j^\dagger(x)\t y \t \L_j^\dagger(z)\\
&\quad + 2x\t \L_j(y) \t \L_j^\dagger(z) + x\t \L_j^2(y) \t z +x\t y\t (\L_j^\dagger)^2(z)\Big).
\end{align*}
\end{proposition}

\begin{proof}
Since we know $\alpha'_t$, we find:
\begin{equation}\label{eq:EFtProof}
\begin{aligned}
1- \E [F_t] 
=& \frac{1}{d}\sum_{k,l=1}^d \E [|\langle e_l, (\id - \alpha'_t)(e_k) \rangle|^2] \\
=& \frac{1}{d}\sum_{k,l=1}^d \E [\delta_{k,l} - 2\delta_{k,l}  \Re \langle e_l, \alpha'_t(e_k) \rangle + |\langle e_l, \alpha'_t(e_k) \rangle|^2] \\
=& \frac{1}{d}\sum_{k,l=1}^d \Big(\delta_{k,l} - 2\delta_{k,l} T_t \Big(\Re f^{[\infty,1]}_{kl} (\unit) \Big) + T_t |f^{[\infty,1]}_{kl}|^2 (\unit)\Big) \\
=& \frac{1}{d}\sum_{k,l=1}^d \Big(\delta_{k,l} - \delta_{k,l} \sum_{n=0}^\infty \frac{t^n}{n!} K^n\Big(f^{[\infty,1]}_{kk} + \overline{f^{[\infty,1]}_{kk}}\Big) (\unit) + \sum_{n=0}^\infty \frac{t^n}{n!} K^n \Big(|f^{[\infty,1]}_{kl}|^2\Big) (\unit)\Big)\\
=&\frac{1}{d}\sum_{k,l=1}^d \Big(\delta_{k,l} - \delta_{k,l} \sum_{n=0}^\infty \frac{t^n}{n!} (\langle e_k, \hat{L}^n(e_k) \rangle + \langle e_k, (\hat{L}^\dagger)^n(e_k)\rangle \\
&+ \sum_{n=0}^\infty \frac{t^n}{n!} \langle e_l\otimes e_k, (\check{L}^{(2)})^m (e_k\otimes e_l) \rangle\Big)\\
=& \frac{1}{d}\sum_{k,l=1}^d
\Big( \delta_{k,l} - \delta_{k,l} \langle e_k, (\rme^{t\hat{L}} +\rme^{t\hat{L}^\dagger})(e_k) \rangle + \langle e_l\otimes e_k, \rme^{t\check{L}^{(2)}}(e_k\otimes e_l) \rangle \Big).
\end{aligned}
\end{equation}
Here the third equality follows from \eqref{eq:CLTexp} and the quotient procedure; the fifth from the Leibniz rule and Proposition \ref{prop:modulus}, noticing that $\overline{f_{kk}}(g) = \langle g e_k, e_k \rangle$ and $L^n \overline{f_{kk}}(\unit) = \langle \hat{L}^n(e_k), e_k \rangle =  \langle e_k, (\hat{L}^\dagger)^n (e_k) \rangle$.

The variance is obtained analogously:
\begin{align*}
\V[F_t] =& \E[F_t^2]- \E [F_t]^2 = \E[(1-F_t)^2]- \E[1-F_t]^2\\
=& \frac{1}{d^2}\sum_{i,j,k,l=1}^d \E [|\langle e_i, (\id - \alpha'_t)(e_j) \rangle|^2|\langle e_l, (\id - \alpha'_t)(e_k) \rangle|^2] \\
&- \Big( \frac{1}{d}\sum_{k,l=1}^d \E [|\langle e_l, (\id - \alpha'_t)(e_k) \rangle|^2]\Big)^2\\
=& \frac{1}{d^2}\sum_{i,j,k,l=1}^d \E\Big[ \delta_{k,l}\delta_{i,j}
- 2 \delta_{i,j}\delta_{k,l} \langle e_l, (\alpha'_t+\alpha'^\dagger_t)(e_k) \rangle\\ 
&\qquad +2 \delta_{i,j} \langle e_l\t e_k, (\alpha'_t\t\alpha'^\dagger_t)(e_k\t e_l) \rangle\\ 
&\qquad+  \delta_{i,j}\delta_{k,l} \langle e_i\t e_k, (\alpha'_t\t\alpha'_t+\alpha'^\dagger_t\t\alpha'_t+\alpha'_t\t\alpha'^\dagger_t
+\alpha'^\dagger_t\t\alpha'^\dagger_t)
(e_i\t e_k) \rangle \\
&\qquad - 2 \delta_{i,j} \langle e_i\t e_l\t e_k, ((\alpha'_t+\alpha'^\dagger_t)\t\alpha'_t\t\alpha'^\dagger_t)(e_i\t e_k\t e_l) \rangle \\
&\qquad+ \langle e_j\t e_i\t e_l\t e_k, (\alpha'_t\t\alpha'^\dagger_t\t\alpha'_t\t\alpha'^\dagger_t)(e_i\t e_j\t e_k\t e_l) \rangle\Big]\\
&-\frac{1}{d^2}\Big(\sum_{k,l=1}^d \Big(\delta_{k,l} - 2\delta_{k,l}  \langle e_l, (\rme^{t\hat{L}}+\rme^{t\hat{L}^\dagger})(e_k) \rangle + \langle e_l\otimes e_k, \rme^{t\check{L}^{(2)}}(e_k\otimes e_l) \rangle \Big)\Big)^2.
\end{align*}
The terms in the first sum are all $0,1,2,3,4$-(anti-)linear expressions, respectively, of the type investigated in Propositions \ref{prop:higher-mom} and \ref{prop:modulus}. Following the proof there, we have
\begin{align*}
\E[\langle e_i\t e_k, (\alpha'^\dagger_t\t\alpha'_t)(e_i\t e_k) \rangle]
=& \E[ \overline{f_{ii}^{[\infty]}}f_{kk}^{[\infty]} \circ\alpha'_t]\\
=& \sum_{n=0}^\infty \frac{t^n}{n!} L^n \big(\overline{f_{ii}^{[\infty]}}f_{kk}^{[\infty]} \big)(\unit)\\
=& \sum_{n=0}^\infty \frac{t^n}{n!} 
\langle e_i\t e_k, (\check{L}^{(1,2)})^n(e_i\t e_k) \rangle
\end{align*}
with
\begin{align*}
\check{L}^{(1,2)}(x\t y) =& 
\bar{\L}^\dagger(x)\t y + x\t \bar{\L}(y)\\
&+\frac{\tau}{|J|}\sum_{j\in J} (\L_j^\dagger)^2(x)\t y + 2\L_j^\dagger(x)\t \L_j(y) + x\t \L_j^2(y)
\end{align*}
because
\begin{align*}
L \big(\overline{f_{ii}^{[\infty]}}f_{kk}^{[\infty]} \big)(\unit)
=& (D_{\bar{\L}}\overline{f_{ii}^{[\infty]}}) f_{kk}^{[\infty]}(\unit) 
+ \overline{f_{ii}^{[\infty]}} (D_{\bar{\L}}f_{kk}^{[\infty]})(\unit)\\
&+ \frac{\tau}{|J|} \sum_{j\in J} (D_{\L_j}^2\overline{f_{ii}^{[\infty]}})f_{kk}^{[\infty]}(\unit)
+ 2 (D_{\L_j}\overline{f_{ii}^{[\infty]}})(D_{\L_j}f_{ii}^{[\infty]})(\unit)
+ \overline{f_{kk}^{[\infty]}} (D_{\L_j}^2f_{kk}^{[\infty]})(\unit)\\
=& \frac{\rmd}{\rmd t}\overline{f_{ii}^{[\infty]}(\rme^{t\bar{\L}})} f_{kk}^{[\infty]}(\unit) 
+ \overline{f_{ii}^{[\infty]}(\unit)} \frac{\rmd}{\rmd t}f_{kk}^{[\infty]}(\rme^{t\bar{\L}})\\
&+ \frac{\tau}{|J|} \sum_{j\in J} \frac{\rmd^2}{\rmd t\rmd s} \overline{f_{ii}^{[\infty]}(\rme^{s\L_j}\rme^{t\L_j})}f_{kk}^{[\infty]}(\unit)\\
&\qquad + 2\frac{\rmd^2}{\rmd t\rmd s} \overline{f_{ii}^{[\infty]}(\rme^{s\L_j})} f_{ii}^{[\infty]}(\rme^{t\L_j})
+ \overline{f_{kk}^{[\infty]}(\unit)} \frac{\rmd^2}{\rmd t\rmd s}f_{kk}^{[\infty]}(\rme^{s\L_j}\rme^{t\L_j}) \restriction_{s=t=0}\\
=& \langle e_i\t e_k, \bar{\L}^\dagger(e_i)\t e_k + e_i\t \bar{\L}(e_k)\rangle\\
&+ \frac{\tau}{|J|} \sum_{j\in J} \langle e_i\t e_k, (\L_j^\dagger)^2(e_i)\t e_k  + 2 \L_j^\dagger(e_i)\t \L_j(e_k) + e_i\t \L_j^2(e_k)\rangle\\
=& \langle e_i\t e_k, \check{L}^{(1,2)}(e_i\t e_k) \rangle.
\end{align*}
For the other 2-(anti-)linear expressions we obtain similar results but with operators $\check{L}^{(1,1)},\check{L}^{(1,1)\dagger},\check{L}^{(1,2)\dagger}$ instead. The remaining terms are treated analogously, by letting $L$ act on the corresponding $m$-(anti-)linear functions, \eg the 3-(anti-)linear case is obtained writing
\begin{align*}
\E [ (f_{ii}^{[\infty]}+\overline{f_{ii}^{[\infty]}}) f_{kl}^{[\infty]} \overline{f_{kl}^{[\infty]}} \circ \alpha'_t] 
=& \sum_{n=0}^\infty \frac{t^n}{n!}L^n ((f_{ii}^{[\infty]}+\overline{f_{ii}^{[\infty]}}) f_{kl}^{[\infty]} \overline{f_{kl}^{[\infty]}})\\
=& \sum_{n=0}^\infty \frac{t^n}{n!}
\langle e_i\t e_l\t e_k, ((\check{L}^{(3,1)})^n+ (\check{L}^{(3,2)})^n)(e_i\t e_k\t e_l \rangle.
\end{align*}
Putting together all of this and expressing the power series back again as exponential functions, we finally obtain the statement in the proposition:
\begin{align*}
\V[F_t] =&\frac{1}{d^2}\sum_{i,j,k,l=1}^d \E\Big[ \delta_{k,l}\delta_{i,j}
- 2 \delta_{i,j}\delta_{k,l} \langle e_l, (\alpha'_t+\alpha'^\dagger_t)(e_k) \rangle\\ 
&\qquad +2 \delta_{i,j} \langle e_l\t e_k, (\alpha'_t\t\alpha'^\dagger_t)(e_k\t e_l) \rangle\\ 
&\qquad+  \delta_{i,j}\delta_{k,l} \langle e_i\t e_k, (\alpha'_t\t\alpha'_t+\alpha'^\dagger_t\t\alpha'_t+\alpha'_t\t\alpha'^\dagger_t
+\alpha'^\dagger_t\t\alpha'^\dagger_t)
(e_i\t e_k) \rangle \\
&\qquad - 2 \delta_{i,j} \langle e_i\t e_l\t e_k, ((\alpha'_t+\alpha'^\dagger_t)\t\alpha'_t\t\alpha'^\dagger_t)(e_i\t e_k\t e_l) \rangle \\
&\qquad+ \langle e_j\t e_i\t e_l\t e_k, (\alpha'_t\t\alpha'^\dagger_t\t\alpha'_t\t\alpha'^\dagger_t)(e_i\t e_j\t e_k\t e_l) \rangle\Big]\\
&-\frac{1}{d^2}\Big(\sum_{k,l=1}^d \Big(\delta_{k,l} - 2\delta_{k,l}  \langle e_l, (\rme^{t\hat{L}}+\rme^{t\hat{L}^\dagger})(e_k) \rangle + \langle e_l\otimes e_k, \rme^{t\check{L}^{(2)}}(e_k\otimes e_l) \rangle \Big)\Big)^2\\
=& \frac{1}{d^2}\sum_{i,j,k,l=1}^d \Big( \delta_{k,l}\delta_{i,j} 
- 2 \delta_{i,j}\delta_{k,l} \langle e_l, (\rme^{t\hat{L}}+\rme^{t\hat{L}^\dagger})(e_k) \rangle \\
&\qquad+2 \delta_{i,j} \langle e_l\t e_k, \rme^{t\check{L}^{(2)}}(e_k\t e_l) \rangle\\ 
&\qquad+  \delta_{i,j}\delta_{k,l} \langle e_i\t e_k, (\rme^{t\check{L}^{(1,1)}}+\rme^{t\check{L}^{(1,2)}} +\rme^{t\check{L}^{(1,2)\dagger}}+\rme^{t\check{L}^{(1,1)\dagger}})
(e_i\t e_k) \rangle \\
&\qquad - 2 \delta_{i,j} \langle e_i\t e_l\t e_k, (\rme^{t\check{L}^{(3,1)}}+\rme^{t\check{L}^{(3,2)}})(e_i\t e_k\t e_l) \rangle \\
&\qquad+ \langle e_j\t e_i\t e_l\t e_k, \rme^{t\check{L}^{(4)}}(e_i\t e_j\t e_k\t e_l) \rangle\Big)\\
&-\frac{1}{d^2}\Big(\sum_{k,l=1}^d \Big(\delta_{k,l} - 2\delta_{k,l}  \langle e_l, (\rme^{t\hat{L}}+\rme^{t\hat{L}^\dagger})(e_k) \rangle + \langle e_l\otimes e_k, \rme^{t\check{L}^{(2)}}(e_k\otimes e_l) \rangle \Big)\Big)^2.
\end{align*}
\end{proof}

For comparison reasons and some applications in \cite{ABH}, we would like to state the analogous formulae for the case of the drift-like continuous-time limit in the sense of Remark \ref{rem:drift-limit}. Since $L^{(\drift)}$ can be regarded as a special case of $L$ with vanishing $\L_j$, the expressions in Proposition \ref{prop:EFtVarFt} simplify significantly and we obtain:

\begin{proposition}\label{prop:EFtVarFtdrift}
In the setting of Proposition \ref{prop:EFtVarFt} but with $L^{(\drift)}$ instead of $L$, we obtain
\[
\E^{(\drift)}[F_t] = 1- \frac{1}{d}\sum_{k,l=1}^d |\langle e_l, (\id - \rme^{t\hat{L}^{(\drift)}})(e_k) \rangle |^2
\]
and $\V^{(\drift)}[F_t] = 0$.
\end{proposition}
Some readers might find the vanishing variance intuitively expected, given that the limiting procedure corresponds somehow to the classical law of large numbers where convergence is almost surely to the (non-constant but time-dependent) expectation value. 

\begin{proof}
The expression for $\E^{(\drift)}$ follows immediately from that of $\E$ in the preceding proof, specialising to $\hat{L}^{(\drift)}$:
since
\[
\langle e_l\t e_k, \rme^{t (\bar{\L}\t \id + \id \t \bar{\L}^\dagger)}(e_k\t e_l)\rangle = \langle e_l \t \rme^{t \bar{\L}}(e_k), \rme^{t \bar{\L}}(e_k) \t e_l \rangle,
\]
the last line in \eqref{eq:EFtProof} becomes simply
\[
\langle e_l\t (\id + \rme^{t \bar{\L}})(e_k), (\id+\rme^{t \bar{\L}})(e_k)\t e_l)\rangle = |\langle e_l\t (\id + \rme^{t \bar{\L}})(e_k)|^2.
\] 

For $\V^{(\drift)}$, we analogously compute:
\begin{align*}
\V^{(\drift)}[F_t] =& \E^{(\drift)}[F_t^2]- \E^{(\drift)} [F_t]^2 = \E^{(\drift)}[(1-F_t)^2]- \E^{(\drift)} [1-F_t]^2\\
=& \frac{1}{d^2}\sum_{i,j,k,l=1}^d \E^{(\drift)} [|\langle e_i, (\id - \alpha'_t)(e_j) \rangle|^2|\langle e_l, (\id - \alpha'_t)(e_k) \rangle|^2] \\
&- \Big( \frac{1}{d}\sum_{k,l=1}^d \E^{(\drift)} [|\langle e_l, (\id - \alpha'_t)(e_k) \rangle|^2]\Big)^2\\
=& \frac{1}{d^2}\sum_{i,j,k,l=1}^d |\langle e_i, (\id - \rme^{t\hat{L}^{(\drift)}})(e_j) \rangle|^2 |\langle e_l, (\id - \rme^{t\hat{L}^{(\drift)}})(e_k) \rangle|^2 \\
&- \Big(\frac{1}{d}\sum_{k,l=1}^d |\langle e_l, (\id - \rme^{t\hat{L}^{(\drift)}})(e_k) \rangle|^2\Big)^2\\
=& 0.
\end{align*}
\end{proof}

\begin{example}
We return to our former illustrative Example \ref{ex:1}. Since $\L= \rmi\Ad(H)$, we find that $\L^\dagger=-\L$ and hence 
\[
\hat{L}^\dagger= \hat{L}, \quad \check{L}^{(2)} = \hat{L}^{(2)},
\]
as follows immediately from the respective definition in Propositions \ref{prop:higher-mom} and \ref{prop:modulus}.
For short times $t$ the results of Proposition \ref{prop:EFtVarFt} become:
\begin{align*}
1-\E[F_t] \approx& -\frac{1}{4^N}\sum_{k}^{4^N} 2 t \langle e_k, \hat{L}(e_k)\rangle 
+\frac{1}{4^N}\sum_{k,l}^{4^N} t \langle e_k\t e_l, \hat{L}^{(2)}(e_l\t e_k)\rangle \\
=& \frac{2t\tau}{4^N|J|} \sum_{j\in J} \sum_{k,l}^{4^N} |\langle e_k, [v_j H v_j, e_l]\rangle|^2\\
=& \frac{2t\tau}{4^N|J|} \sum_{j\in J} \sum_{k}^{4^N} |[v_j H v_j, e_k]|^2,
\end{align*}
which for special cases of $H$ can be further simplified, but in general will be used in this form for a computer and is of order $O(\tau t \|H\|)$. A similar procedure may be applied to variance.

In contrast, in the case of the drift-like limit, we would simply get
\[
\E^{(\drift)}[F_t]=1, \quad \V^{(\drift)}[F_t] = 0, \quad t\in\R_+,
\]
which is obviously less realistic than the diffusion-like limit, but on the other hand confirms that for unitary time-evolution $(\alpha_t)_{t\in\R_+}$ dynamical decoupling works (\ie decouples) optimally, in contrast to other types of $\alpha$, \cf also Theorem \ref{lem:deccond-unitary}! 
\end{example}

\begin{remark}[Errors]
Theorem \ref{th:cont-limit} gives us the expectation of our quantities in the continuum limit, but we must ask two questions:-
\begin{itemize}
\item[(1)] How big is the difference between continuum limit and original discrete random paths?
\item[(2)] What is the distribution of the actual (continuum limit) paths around the expectation value?
\end{itemize}
These two errors add up to give the total maximal error, which we have to estimate now.

Concerning (1), one has to work with a kind of Berry-Esseen theorem \cite{Fel} on the approximation of random walks by Brownian motion. This is quite complicated, but we satisfy ourselves here with the fact that this error tends to $0$ as $\tau\ra 0$.

Concerning (2), the deviation around the expectation value is expressed in the quantiles, which can be efficiently estimated using Chebychev's inequality \cite{Fel} together with the variance expression.
\end{remark}

\section{Estimates and application: intrinsic/extrinsic decoherence}\label{sec:five}

Suppose a given quantum system $(\H,\L)$ undergoes decoherence caused by interaction with an external quantum heat bath described by another quantum system $(\H_1, \rmi \ad(H_1))$. Then according to standard axioms of quantum mechanics, time evolution of the total (closed) system is unitary, thus described by a one-parameter automorphism family on the operators of the total Hilbert space $\H'=\H\otimes\H_1$, namely
\[
t\in\R_+ \mapsto \T\rme^{\rmi \int_0^t \ad (H'(t'))\rmd t'},
\]
and $H'$ is the (possibly time-dependent) Hamiltonian of the total system on $\H'$ and the time-ordered integral is defined in analogy to \eqref{eq:time-order}.
The heat bath may be infinite-dimensional separable, but the involved Hamiltonian $H'$ is henceforth supposed to be uniformly bounded on compact intervals. It is unclear whether dynamical decoupling works without this assumption, and maybe alternative requirements would have to be made in case of unboundedness, \cf \cite[App.]{ABH} for further discussion.

The actual dynamics perceived on the subsystem $\H$ is given by
\[
t\in\R_+ \mapsto \alpha_t := \EE\circ \T \rme^{\rmi \int_0^t \ad (H'(t'))\rmd t'}(\cdot\otimes\rho^\theta),
\]
where $\EE:B(\H')\ra B(\H)$ is the partial trace (conditional expectation) onto the subsystem and $\rho^\theta$ the initial state of the heat bath \cite{LB,pb}. The resulting perceived dynamics $(\alpha_t)_{t\in\R_+}$ then becomes a family of CPT maps. Under special assumptions on $H'$, it actually produces the CPT semigroup with infinitesimal generator $\L$, the Lindblad operator, but usually $\alpha_t$ is no longer an automorphism. We call this phenomenon, where a CPT semigroup time evolution arises from interaction with an external quantum heat bath and unitary time evolution on the total system, \emph{extrinsic decoherence} because the non-unitarity of time evolution of the original system is caused by interaction with the external heat bath.

In contrast to this, \emph{intrinsic decoherence} we call the situation where time evolution of a \emph{closed} system $(\H,\L)$ is no longer unitary and the non-unitarity is intrinsic to the system, \ie does not arise from (unitary) interaction with a heat bath. It is a fundamental question whether this actually occurs in nature or whether the axiom of unitarity is always fulfilled -- on a sufficiently large total system.
Mathematically the two cases are described in the same way (by CPT semigroups with unitary dilations), and also physically with usual observations they seem to be indistinguishable. 

However, applying dynamical decoupling in the case of the above type of extrinsic decoherence, the time evolution of the total system is unitary, and so the perceived evolution on the subsystem is given by the discrete stochastic process
\[
\alpha^{(\tau)}_{n\tau} = \EE \circ \prod_{i=1}^n \Ad(v_{j_i}\otimes\unit) \circ \T \rme^{\rmi \int_{(i-1)\tau}^{i\tau} \ad (H'(t'))\rmd t'} \circ \Ad(v_{j_i}^*\otimes \unit)(\cdot\otimes\rho^\theta).
\]
Now we notice that, if \eqref{eq:deccond1} is satisfied for all $x\in B(\H)$, then it is also satisfied for all $x\in B(\H')$ modulo $\unit\otimes B(\H_1)$. In fact, $x\in B(\H')$ can be written as a finite sum $\sum_k y_k\otimes z_k + \unit\otimes \tilde{z}$, with certain traceless $y_k\in B(\H)$ and with $z_k,\tilde{z}\in B(\H_1)$, and then
\[
\frac{1}{|J|}\sum_{j\in J}\sum_{k} 
(v_j\otimes \unit)(y_k\otimes z_k)(v_j\otimes \unit)^* +\unit \otimes \tilde{z}
= \frac{1}{|J|}\sum_k \Big(\sum_{j\in J} v_j y_k v_j^*\Big) \otimes z_k + \unit \otimes \tilde{z}
=  \unit \otimes \tilde{z}.
\]

Consider now for $x$ the (possibly time-dependent) Hamiltonian $H'=\sum_k H_{0,k}\otimes H_{1,k} + \unit\otimes H_1$. The heat bath is by definition in a thermal equilibrium state $\rho^\theta$ independent of time, \ie $\ad(H_1)(\rho^\theta)=0$. Let $\L':=\rmi \ad(H')$ be the (purely unitary) Lindbladian of the total system and thus $\bar{\L}'=\rmi\ad(\unit\t H_1)$, so $\bar{\L}(x\otimes\rho^\theta)=0$, for all $x\in A$. Then we obtain
\[
\hat{L}'= \bar{\L}' + \frac{\tau}{|J|}\sum_{j\in J} \Ad(v_j\otimes\unit) \circ (\L'-\bar{\L}')^2 \circ \Ad(v_j^*\otimes\unit),
\]
and hence 
\[
\E [\rho_t] = \EE \circ \T\rme^{\int_0^t\hat{L}'(t')\rmd t'}(\rho_0\otimes \rho^\theta).
\]
The main dynamics comes from $\bar{\L}'$, which leaves $\rho\otimes\rho^\theta$ invariant, but $\L_j'$ changes it, so that higher-order terms disturb the invariance of the state.

We can conclude: if the system dynamics is determined by extrinsic decoherence then the decoupling condition is satisfied in first-order approximation and the total time evolution under decoupling in first-order approximation in $\tau$ and $t$ is described as in the unitary case; $\unit\t\rho^\theta$ will in general not be invariant under $\L_j'$, but those effects are of higher order in $\tau$.


We would like to have an estimate of $\E[F_t]$ that depends only on $t$ and the coupling strength, distinguishing between the two extremal cases of purely extrinsic decoherence (\ie purely unitary on the dilation: $\Psi'=0$ and $a'=-a'^*$ in the notation of Theorem \ref{lem:deccond-unitary}) and purely intrinsic decoherence ( $\Psi'\not=0$ and $a=a^*$).

\begin{theorem}\label{th:EFt-bounds}
Given a quantum system $(\H,\L)$ with decoupling set $V$ and the previous notation, write $\Gamma:=\max\{\|\L\|,\|\L'\|,\|\bar{\L}\|\}$. Then in the drift-like limit of Remark \ref{rem:drift-limit}, for purely extrinsic decoherence we have 
\[
\E^{(\drift)}[F_t^{(extr)}] = 1.
\]
An approximate upper bound for the expectation of the fidelity in the case of purely intrinsic decoherence, in the limit of $\tau\ll t\ll 1/\Gamma$, is asymptotically given by
\[
\E^{(\drift)}[F_t^{(intr)}] \lesssim 1 - \frac{1}{d}t^2 \|\hat{L}^{(\drift)}\|^2
\]
If in addition $\L=\L^\dagger$ (so-called purely intrinsic dephasing) this can be made more precise:
\[
\E^{(\drift)}[F_t^{(intr)}] \le 1 - \frac{1}{d}(1- \rme^{-t \|\L\|/|J|})^2.
\]
In the (physically more realistic) diffusion-like limit of Theorem \ref{th:cont-limit}, a lower bound for the fidelity of purely extrinsic decoherence is asymptotically given by
\[
\E[F_t^{(extr)}] \gtrsim 1 -\frac{2d}{|J|} \tau \int_0^t\|\L'_0(t')\|^2\rmd t',
\]
in terms of the (possibly time-dependent) Lindbladian $\L'$ on the dilated system, while an upper bound in the case of purely intrinsic decoherence is asymptotically given by
\[
\E[F_t^{(intr)}] \lesssim 1- \frac{2}{d\, |J|} \tau t \|\L-\bar{\L}\|^2 - \frac{1}{d} t^2 \|\bar{\L}\|^2.
\]
\end{theorem}

\begin{proof}
We appeal to Proposition \ref{prop:EFtVarFt} for notation and the exact formulae underlying our estimates here.

For the case of pure dephasing (\ie $\L^\dagger=\L$), we first notice that $\hat{L}^{(\drift)}$, being a sum of double commutators, is a selfadjoint operator on the Hilbert space $A$. Moreover, it must be negative since $\rme^{t\hat{L}^{(\drift)}}$ is a contraction. We may suppose that the orthonormal basis $(e_k)_{k=1\ldots d}$ consists of the eigenvectors of $\hat{L}^{(\drift)}$ in decreasing order of eigenvalues. 
The smallest eigenvalue of $\hat{L}^{(\drift)}$ is less than $-\|\L\|/|J|$: in fact, the smallest eigenvalue of $\L$ is $-\|\L\|$, and $\hat{L}^{(\drift)} \le \frac{1}{|J|}\L$. Then we find
\begin{align*}
1-\E^{(\drift)}[F_t] =& \frac{1}{d}\sum_{k,l=1}^d \langle e_l, (\id - \rme^{t\hat{L}^{(\drift)}})(e_k) \rangle^2\\
\ge& \frac{1}{d} \langle e_1, (\id - \rme^{t\hat{L}^{(\drift)}})(e_1) \rangle^2
= \frac{1}{d} (1- \rme^{-t\|\L\|/|J|})^2.
\end{align*}

For general $\L$, we can say at least
\begin{align*}
1-\E^{(\drift)}[F_t] 
=& \frac{1}{d}\sum_{k=1}^d \langle (\id - \rme^{t\hat{L}^{(\drift)}})(e_k) , (\id - \rme^{t\hat{L}^{(\drift)}})(e_k) \rangle\\
\ge& \frac{1}{d} \| \id - \rme^{t\hat{L}^{(\drift)}} \|^2 \approx \frac{1}{d} t^2 \|\hat{L}^{(\drift)}\|^2.
\end{align*}

Let us come to the actual diffusion-like limit. Using the power series of the exponential function and neglecting higher order terms, we obtain
\begin{align*}
1- \E[F_t] =&
 \frac{1}{d}\sum_{k,l=1}^d \Big( \delta_{k,l} -\delta_{k,l} \langle e_l, (\rme^{t\hat{L}} +\rme^{t\hat{L}^\dagger})(e_k) \rangle + \langle e_l\otimes e_k, \rme^{t\check{L}^{(2)}}(e_k\otimes e_l) \rangle \Big)\\
 \approx& t\Big( \frac{\tau}{d\, |J|} \sum_{k,l=1}^d \sum_{j\in J} 
\langle e_l\otimes e_k, (\L_j\t\L_j^\dagger + \L_j^\dagger\t\L_j)(e_k\otimes e_l) \rangle\Big)\\
&+t^2 \Big( \frac{1}{d}\sum_{k,l}^d \langle e_l\otimes e_k, (\bar{\L}\t\bar{\L}^\dagger)(e_k\otimes e_l) \rangle 
\Big)\\
=& t\Big( \frac{\tau}{d\, |J|} \sum_{j\in J}
2 \tr_{A\t A} (\L_j\t\L_j^\dagger \circ \phi)\Big)\\
&+t^2 \Big( \frac{1}{d} \tr_{A\t A} ( \bar{\L}\t\bar{\L}^\dagger \circ\phi)\Big)\\
\geq & \frac{2}{d\, |J|} \tau t \|\L_j\|^2 + \frac{1}{d} t^2 \|\bar{\L}\|^2, 
\end{align*}
where $j\in J$ is arbitrary, $\phi$ denotes the flip unitary on $A\t A$, and $\t_{A\t A}$ the standard (non-normalised) trace on $B(A\t A)$.

For the case of extrinsic decoherence, we have to work in the dilation algebra $A\t B(\H_1)$. Let $\L'$ be the corresponding (possibly time-dependent) Lindbladian on that algebra, corresponding to unitary time evolution. Then after decoupling we obtain
\[
\hat{L}' = \bar{\L}' + \frac{\tau}{|J|} \sum_{j\in J} (\L'_j)^2, \quad \L'_j:= \Ad(v_j\t\unit)\circ (\L'-\bar{\L}') \circ \Ad(v_j^*\t\unit),
\]
with the commutator $\bar{\L}'$ vanishing on $A\t\rho^\theta$. As in the preceding case, we find
\begin{align*}
1- \E[F_t]=&
\frac{1}{d}\sum_{k,l=1}^d \T\Big(\delta_{k,l} -\delta_{k,l} \langle e_l\t\unit, (\rme^{\int_0^t \hat{L}'(t')\rmd t'} +\rme^{\int_0^t \hat{L}'^\dagger (t')\rmd t'})(e_k\t\rho^\theta) \rangle \\
&\quad + \langle e_l\t\unit\t e_k\t\unit, \rme^{\int_0^t\check{L}'^{(2)}(t')\rmd t'}(e_k\t\rho^\theta\t e_l\t\rho^\theta) \rangle \Big)\\
\approx& \frac{\tau}{d\, |J|} \sum_{k,l=1}^d \sum_{j\in J} 
\Big\langle e_l\t\unit\t e_k\t\unit, \int_0^t(\L_j'\t\L_j'^\dagger + \L_j'^\dagger\t\L_j')(t')\rmd t' (e_k\t \rho^\theta\t e_l\t\rho^\theta) \Big\rangle\\
=& \frac{\tau}{d\, |J|} \sum_{j\in J}
2 \tr_{A\t A\t B(\H_1) \t B(\H_1)} \Big( \int_0^t(\L_j'\t\L_j'^\dagger)(t')\rmd t' \circ \phi^{\t 2} (\unit\t\rho^\theta\t\unit\t\rho^\theta)\Big)\\
\leq & \, 2d \tau \int_0^t\|\L'_0(t')\|^2\rmd t', 
\end{align*}
because $\tr_{A\t B(\H_1)}(\unit\t\rho^\theta)=d$ and $\bar{\L}'(x\otimes\rho^\theta)=0$ for all $x\in A$.
\end{proof}

Notice that these bounds are probably not sharp at all, but they should rather serve as an inspiration and starting point for finding more specific and sharper bounds. The interesting fact, in any case, is that they separate the fidelity of intrinsic and extrinsic decoherence dynamics, respectively, in the region $\tau\ll t\ll 1/\Gamma$, \cf Figure \ref{fig1}.

Moreover, under further assumptions on $\L$ like \eg $\L^\dagger=\L$ we can try to use a similar procedure in order to achieve better bounds involving directly $\|\L\|$.

\begin{figure}
\centering
\includegraphics[width=10cm]{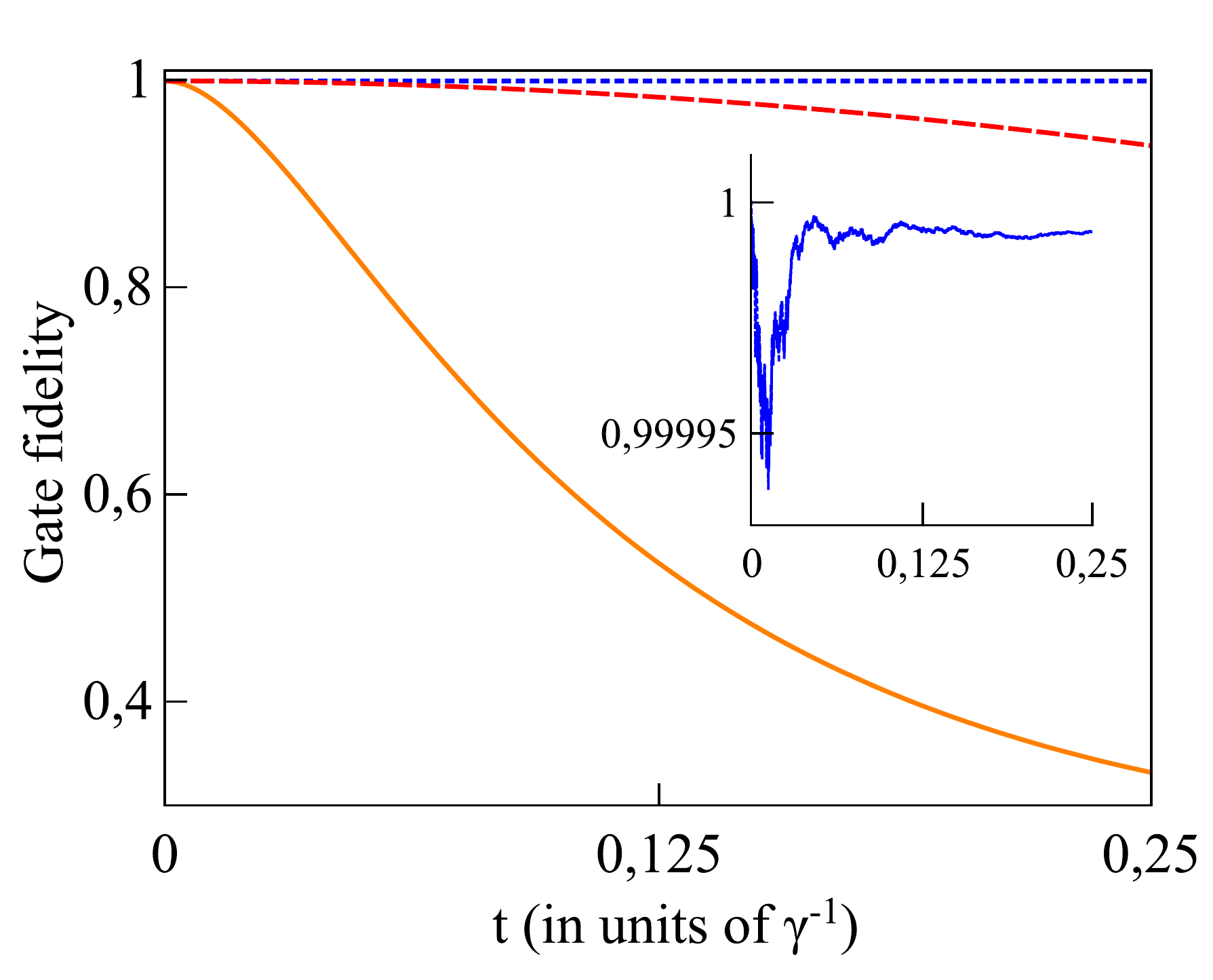}
\caption{Numerical evaluation of average $\overline{F_t^{ext}}$ (dotted blue line) coinciding with the lower bound for $\E^{(\drift)}[F_t^{ext}]\equiv 1$, average $\overline{F_t^{int}}$ (solid orange line), and the upper bound for $\E^{(\drift)}[F_t^{int}]$ from Theorem \ref{th:EFt-bounds} (dashed red line) as a function of $t$, for $\Gamma\tau=10^{-3}$ and the so-called amplitude-damping channel model with coupling strength $\gamma$ (\cf Example \ref{ex:2}(2) and \cite{ABH} for further explanation). The average was taken over 25 paths, one of them illustrated for  $F_t^{ext}$ in the inset plot.}\label{fig1}
\end{figure}

\begin{conclusion}[Application to experiment -- sketch]
If $\tau$ is sufficiently small, then for suitable $t$ the last two bounds in Theorem \ref{th:EFt-bounds} provide a separation into two disjoint ranges of the fidelity in the two intrinsic interaction cases, which allows the experimenter to identify the type of decoherence. He would have to proceed as follows:
\begin{itemize}
\item[(1)] Given the intrinsic or extrinsic coupling strength $\Gamma:=\max\{\|\L\|,\|\L'\|,\|\bar{\L}\|\}$, choose $t\ll 1/\Gamma$.
\item[(2)] Choose and vary $\tau\ll t$ in that range.
\item[(3)] Compute the fidelity of many decoupling pulse sample paths for these given values of $\Gamma$ and varying $\tau,t$, then average and extrapolate to get his averaged fidelity $\bar{F}_t$ as a function of $\tau$ and $t$.
\item[(4)] Compare it with the bounds for these given values of $\Gamma,\tau$: then in the above limit, he will find either
\[
\bar{F}_t \gtrsim 1 - 2d \tau \int_0^t\|\L'_0(t')\|^2\rmd t' \quad \mbox{or} \quad 
\bar{F}_t \lesssim 1- \frac{2}{d\, |J|} \tau t \|\L-\bar{\L}\|^2 - \frac{1}{d} t^2 \|\bar{\L}\|^2,\quad t\in\R_+.
\]
the first case corresponding to extrinsic, the second to intrinsic dephasing.
\end{itemize}

If he cannot carry out many runs, then it would also be necessary to take into account the quantiles from above in order to understand how well his experimental mean value $\bar{F}_t$ describes the analytical $\E[F_t]$. This can be done by considering higher moments $\E[F_t^n]$ and $\V[F_t]$ as in Proposition \ref{prop:EFtVarFt} For an arbitrarily large number of runs, however, this is not necessary. Figure \ref{fig1} illustrates the bounds with an average of concrete sample paths.

Although the precise relation between $\|\L\|, \|\L_0\|, \|\bar{\L}\|, \|\L'\|$ etc.~is not clear and it is therefore difficult to compare the above bounds quantitatively, these values depend somehow monotonically on one another, \ie increase or decrease synchronically. In any case, when $\tau\ra 0$ then the higher-order terms can be neglected and the difference between discrete random walk and continuum-limit tends to $0$. At this point boundedness of $H'$ is needed; otherwise alternative assumptions would have to be made that lead to future work, \cf also \cite{ABH}. Moreover, it follows from the above bounds then that $\bar{F}_t\ra 1$ in the extrinsic case, whereas $\bar{F}_t$ converges to some function $f(t)\lesssim 1-t^2\|\bar{\L}\|^2/d$ in the intrinsic case.

If extrinsic and intrinsic decoherence appear together: since the respective coupling strengths will not be known, it is impossible to compute the above bounds; yet, for fixed $t$, letting $\tau\ra 0$, the experimenter can check whether or not $\bar{F}_t$ goes to $1$, meaning pure extrinsic or intrinsic/mixed decoherence, respectively. \boxy
\end{conclusion}

\appendix

\section{Lie groups and convolution semigroups}\label{sec:app}

The aim of this appendix section is to sketch the necessary definitions and facts about Lie groups and convolution semigroups necessary to understand the third and fourth section. For a comprehensive study and the notation used in Section \ref{sec:three} we refer to any suitable textbook: \eg \cite{FH} for (linear) Lie groups and algebras, \cite{da,EN} for one-parameter semigroups, \cite[Ch.4]{gr,he} for probability and convolution measures on Lie groups.

An \emph{$N$-dimensional (real) Lie group} $G$ is an $N$-dimensional smooth manifold which has a group structure with neutral element $\mathbf{1}$ such that multiplication and inversion are smooth maps. In this paper $G$ is always a \emph{linear algebraic Lie group}, \ie a group of linear mappings on a finite-dimensional real vector space. 
The tangent space $T_\mathbf{1} G$ of $G$ in $\mathbf{1}$ forms a Lie algebra and is called the \emph{Lie algebra of $G$}, denoted $\mathfrak{g}$ with scalar product (nondegenerate bilinear form) $\langle\cdot,\cdot\rangle_\mathfrak{g}$. 

There is a canonical diffeomorphism $\exp$ from a $0$-neighbourhood in $\mathfrak{g}$ to some $\mathbf{1}$-neighbourhood $U\subset G$ mapping $0$ to $\mathbf{1}$ and called the exponential map, which in the present case can be identified with the standard exponential function of matrices. Given an orthonormal basis $(X_k)_{k=1\ldots N}$ of $\mathfrak{g}$, the corresponding \emph{(coordinate) $\mathbf{1}$-chart} is the smooth function $x: U\ra \mathbb{R}^N$ such that
\[
g= \exp \Big( \sum_{k=1}^N x_k(g) X_k \Big), \quad g\in U,
\]
and $x_k:U\ra\mathbb{R}$ is called the \emph{$k$-th coordinate map}. One may extend the functions $x_k\in C^\infty(U)$ to functions in $C^\infty_c(G)$ denoted again by $x_k$; write $G_c$ for the one-point compactification of $G$ if $G$ is noncompact, otherwise we take $G_c=G$, and every function $f\in C_c(G)$ is extended by $f(\infty):=0$ to $G_c$; this is needed in section \ref{sec:three} for technical reasons.

One notes that
\[
\frac{\rmd}{\rmd t} x_k(\rme^{t Y}) \restriction_{t=0} = \langle Y, X_k\rangle_\mathfrak{g},
\]
for every $k=1,\ldots,N$ and $Y\in\mathfrak{g}$. The directional derivative $D_Y$ for $Y\in\mathfrak{g}$ is defined by
\[
D_Y f(g) := \frac{\rmd}{\rmd t} f(\rme^{t Y}g) \restriction_{t=0}, \quad f\in C^1_c(G), \, g\in G,
\]
and one has $D_{X+\lambda Y} = D_X + \lambda D_Y$, for $X,Y\in\mathfrak{g}$ and $\lambda\in\mathbb{R}$.

The {\em convolution} of two probability measures $\mu_1,\mu_2$ on $G$ (with usual Borel $\sigma$-algebra $\mathfrak{B}(G)$) is  defined by 
\[
\mu_1*\mu_2 (A) := (\mu_1\times\mu_2)\{(g,h)\in G\times G : gh \in A\},
\quad A\in \mathfrak{B}(G).
\]
Suppose that $\mu_1$ and $\mu_2$ are supported in a subsemigroup $H\subset G$. Then for every $A\in \mathfrak{B}(G)$, we have
\begin{align*}
\mu_1*\mu_2 (A) =& (\mu_1\times\mu_2)\{(g,h)\in G\times G : gh \in A\}\\
=& (\mu_1\times\mu_2)\{(g,h)\in H\times H : gh \in A\}\\
=& (\mu_1\times\mu_2)\{(g,h)\in H\times H : gh \in A\cap H\}\\
=&\mu_1*\mu_2 (A\cap H),
\end{align*}
so $\mu_1*\mu_2$ is supported in $H$, too.

Measures and convolution on $G$ can be trivially extended to $G_c$ by setting $g\infty:=\infty g := \infty$, for all $g\in G_c$. The set of probability measures on $G_c$,  equipped with the *-weak topology and convolution as multiplication, constitutes a topological
monoid, where the Dirac measure $\delta _\unit$ serves as the neutral element. Here the {\em *-weak topology} on $G_c$ is defined as
follows: a net of measures $(\mu_i)_{i \in I}$ converges to a
limit measure $\mu$ if for all $f \in C(G_c)$ (the continuous $\R$-valued functions on $G_c$) the condition $\int_{G_c} f
\rmd\mu_i \ra \int_{G_c} f  \rmd\mu$ holds.
A  {\em continuous convolution semigroup of probability measures on $G$} is a set $(\mu_t)_{t \in \R _+}$ of probability measures on $G$ (trivially extended to $G_c$) such that $\mu_s*\mu_t =\mu_{s+t}$ 
for every $s,\ t \in \R _+$ and $\lim_{t\ra0}\mu_t=\mu_0=\delta
_{\unit}$ *-weakly.

Let $(\mu _t)_{t\in\R_+}$ be a continuous convolution semigroup of probability
measures on  $G$. For $t\in \R _+$ define the operator
\[
T_t: C(G_c) \ra C(G_c), \; (T_t f) (g) := \int_{G_c} f(gh) d\mu _t(h),\quad g\in G_c.
\]
Then $(T_t)_{t \in \R _+}$ forms a {\em strongly continuous one-parameter contraction semigroup on $C(G_c)$}.
To $(T_t)_{t\in \R _+}$ there corresponds an {\em infinitesimal generator}
\[
L := \lim_{t \ra 0} \frac{T_t -\id }{t}
\]
(in the strong operator topology) on a suitable $T_t$-invariant dense domain $\dom (L)\subset C(G_c)$.

Given a strongly continuous one-parameter semigroup $(T_t)_{t\in\R_+}$ with generator $(L,\dom(L))$ on a Banach space $E$, we write $\Cci(L):=\bigcap_{n\in\N} \dom(L^n)$. A vector $f\in\Cci(L)$ is called \emph{entire analytic for $L$} if
\[
z\in \C\mapsto \sum_{n=0}^\infty \frac{z^n}{n!} L^n f\in \Cci(L)
\]
is analytic, in which case it extends $t\in\R_+\mapsto T_t f\in E$ to an entire analytic function. Nonzero analytic vectors need not exist for one-parameter semigroups, whereas for one-parameter groups they do.

If $(T_t)_{t\in\R_+}$ is a strongly continuous one-parameter semigroup on a Banach space $E$ with generator $(L,\dom(L))$ which leaves invariant a closed subspace $E_0\subset E$, then it induces a strongly continuous one-parameter semigroup $(S_t)_{t\in\R_+}$ on the quotient Banach space $E/E_0$ with infinitesimal generator $(K,\dom(K))$ as follows: denote the quotient map by $q:E\ra E/E_0$, then
\[
S_t q(f) := q(T_t f), \quad f\in E,
\]
and $K q(f) = q(L f)$ with dense $\dom(K) = q(\dom(L))\subset E/E_0$.

Given a continuous convolution semigroup of probability measures $(\mu_t)_{t \in \R _+}$
on $G$, there exists a probability space and a $G$-valued
Markov process on this space such that its transition
probabilities from $(g,0)\in G_c \times \R _+$ to $(A,t) \in \mathfrak{B}
(G_c)\times \R _+$ are given by $(\mu_t * \delta _{g})(A)$.  
The most interesting processes on $G$
are the ones we encounter in Section \ref{sec:three}, the so-called {\em Gaussian processes}, whose contraction semigroups have generators of the form
\[
L=\sum_{k=1}^N a_k D_{X_k} +\sum_{k,l=1}^N a_{kl} D_{X_k} D_{X_l}, 
\]
with $a_k\in \R$  and $(a_{kl})_{k,l=1...N}$ forms a
positive-definite matrix, and with $\dom(L)=C^2(G_c)$.

\end{document}